\documentclass{article}

\usepackage{graphicx}

\usepackage{amsmath,amssymb,amsthm}

\usepackage{url,mathrsfs}

\theoremstyle{plain}

    \newtheorem{thm}{Theorem}[section]
    \newtheorem{theorem}[thm]{Theorem}

    \newtheorem{corollary}[thm]{Corollary}

\theoremstyle{definition}

    \newtheorem{definition}[thm]{Definition}
    \newtheorem{example}[thm]{Example}

\newcommand{\relstr}[1]{\mathbb{#1}}
\newcommand{\arity}[1]{\mathrm{ar}(#1)}
\newcommand{\mc}[1]{\mathbf{#1}}
\newcommand{\CSP}{\mathrm{CSP}}
\newcommand{\maj}{\mathrm{maj}}
\newcommand{\PCSP}{\mathrm{PCSP}}
\newcommand{\Pol}{\mathrm{Pol}}
\newcommand{\vc}[1]{\mathbf{#1}}
\newcommand{\clone}[1]{\mathcal{#1}}
\newcommand{\class}[1]{\mathscr{#1}}
\newcommand{\obj}[1]{\mathfrak{#1}}
\newcommand{\struct}[1]{\mathfrak{#1}}

\begin{document}
\title{Algebraic Theory of Promise Constraint Satisfaction Problems, First Steps}
%
%
\author{Libor Barto
\thanks{Libor Barto has received funding from the European Research Council
(ERC) under the European Unions Horizon 2020 research and
innovation programme (grant agreement No 771005)}
%
\\ Charles University \\ Faculty of Mathematics and Physics \\ Department of Algebra \\ Sokolovsk\'a 83, 18675 Praha 8, Czechia\\
email: libor.barto@gmail.com}
\date{May 31, 2019}

\maketitle              
\begin{abstract}
What makes a computational problem easy (e.g., in P, that is, solvable in polynomial time) or hard (e.g., NP-hard)? This fundamental question now has a satisfactory answer for a quite
broad class of computational problems, so called fixed-template constraint satisfaction
problems (CSPs) -- it has turned out that their complexity is captured by a certain specific
form of symmetry. This paper explains an extension of this theory to a much broader class of
computational problems, the promise CSPs, which includes relaxed versions of CSPs such as the problem of finding a
137-coloring of a 3-colorable graph.

\end{abstract}

\section{Introduction}

In Computational Complexity we often try to place a given computational problem into some familiar complexity class, such as P, NP-complete, etc. In other words, we try to determine the image of a computational problem under the following mapping $\Phi$.
\begin{align*}
\Phi: \mbox{computational problems} &\to \mbox{complexity classes} \\
  \mbox{problem} &\mapsto \mbox{its complexity class}
\end{align*}
When we try to achieve this goal for a whole class of computational problems, say $\class{S}$, it is a natural idea to look for some intermediate collection $\obj{I}$ of ``invariants'' and a decomposition of $\Phi$ through $\obj{I}$:
$$
\class{S}  \stackrel{\Psi}{\rightarrow} \obj{I} \rightarrow \mbox{complexity classes}
$$
Members of $\obj{I}$ are then objects that exactly capture the computational complexity of problems in $\class{S}$. The larger $\class{S}$ is and the more objects $\Psi$ glues together, the better such a decomposition is.

This idea proved greatly useful for an interesting class of problems, so called fixed-template constraint satisfaction problems (CSPs), and eventually led to a full complexity classification result~\cite{Bul17,Zhu17}. In a decomposition, suggested in~\cite{FV98} and proved in~\cite{Jea98}, $\Psi$ assigns to a CSP a certain algebraic object that describes, informally, the high dimensional symmetries of the CSP. This basic insight of the so called \emph{algebraic approach to CSPs} was later twice improved~\cite{BJK05,BOP18}, giving us a chain
$$
\mbox{CSPs}  \stackrel{\Psi}{\rightarrow} \obj{I}_1 {\rightarrow} \obj{I}_2 {\rightarrow} \obj{I}_3 \rightarrow \mbox{complexity classes}.
$$

The basics of the algebraic theory can be adapted and applied in various generalizations and variants of the fixed-template CSPs, see surveys in~\cite{KZ17}. One particularly exciting direction is a recently proposed significant generalization of CSPs, so called promise CPSs (PCSPs)~\cite{AGH17,BG18}.  
This framework is substantially richer, both on the algorithmic and the hardness side, and a full complexity classification is wide open even in very restricted subclasses. On the other hand, the algebraic basics can be generalized from CSP to PCSP and, moreover, one of the results in~\cite{BKO19} not only gives such a generalization but also provides an additional insight and simplifies the algebraic theory of CSPs. 

The aim of this paper is to explain this result (here Theorem~\ref{thm:forth}) and the development in CSPs leading to it (Theorems~\ref{thm:first}, \ref{thm:second}, \ref{thm:third}). The most recent material comes from the conference papers~\cite{BKO19} and \cite{Bar19}, which will be merged and expanded in~\cite{BBKO}. Very little preliminary knowledge is assumed but an interested reader may find an in depth introduction to the fixed-template CSP and its variants in~\cite{KZ17}.

\section{CSP} \label{sec:CSP}

Fur the purposes of this paper, we define
a \emph{finite relational structure} as a tuple $\relstr{A} = (A; R_1, \dots, R_n)$, where $A$ is a finite set, called the \emph{domain} of $\relstr{A}$, and each $R_i$ is a relation on $A$ of some arity, that is, $R_i \subseteq A^{\arity{R_i}}$ where $\arity{R_i}$ is a natural number. 

A \emph{primitive positive formula (pp-formula)} over $\relstr{A}$ is a formula that uses only existential quantification, conjunction, relations in $\relstr{A}$, and the equality relation. We will work only with formulas in a prenex normal form. 

\begin{definition} \label{def:CSP}
Fix a finite relational structure $\relstr{A}$. The CSP over $\relstr{A}$, written $\CSP(\relstr{A})$, is the problem of deciding whether a given pp-sentence over $\relstr{A}$ is true.

In this context, $\relstr{A}$ is called the \emph{template} for $\CSP(\relstr{A})$.  
\end{definition}

For example, if $\relstr{A} = (A; R,S)$ and both $R$ and $S$ are binary, then an input to $\CSP(\relstr{A})$ is, e.g.,
\[
(\exists x_1 \exists x_2 \dots \exists x_5) \ R(x_1,x_3) \wedge S(x_5,x_2) \wedge R(x_3,x_3).
\]
This sentence is true if there exists a \emph{satisfying assignment}, that is, a mapping $f: \{x_1, \dots, x_5\} \to A$ such that $(f(x_1),f(x_3)) \in R$, $(f(x_5),f(x_2)) \in S$, and $(f(x_3),f(x_3)) \in R$. Each conjunct thus can be thought of as a constraint limiting $f$ and the goal is to decide whether there is an assignment satisfying each constraint. 

Clearly, $\CSP(\relstr{A})$ is always in NP.

The CSP over $\relstr{A}$ can be also defined as a search problem where the goal is to find a satisfying assignment when it exists. It has turned out that the search problem is no harder then the decision problem presented in Definition~\ref{def:CSP}~\cite{BJK05}.

\subsection{Examples}

Typical problems covered by the fixed-template CSP framework are satisfiability problems, (hyper)graph coloring problems, and equation solvability problems. Let us look at several examples. We use here the notation
$$
E_k = \{0,1, \dots,k-1\}.
$$

\begin{example} \label{ex:sat}
Let $3\relstr{SAT} = (E_2; R_{000}, R_{001}, \dots, R_{111})$, where
$$
	R_{abc} = E_2^3 \setminus \{(a,b,c)\} \ \mbox{ for all } a,b,c \in \{0,1\}.
$$
An input to $\CSP(3\relstr{SAT})$ is, e.g., 
$$
(\exists x_1 \exists x_2 \dots) \ S_{001}(x_1,x_4,x_2) \wedge S_{110}(x_2,x_5,x_5) \wedge S_{000}(x_2,x_1,x_3) \wedge \dots.
$$
Observe that this sentence is true if and only if the propositional formula
$$
(x_1 \vee x_4 \vee \neg x_2) \wedge (\neg x_2 \vee \neg x_5 \vee x_5) \wedge (x_2 \vee x_1 \vee x_3) \wedge \dots
$$ 
is satisfiable. Therefore $\CSP(3\relstr{SAT})$ is essentially the same as the 3SAT problem, a well known NP-complete problem.

On the other hand, recall that the 2SAT problem, which is the CSP over $2\relstr{SAT} = (E_2; R_{00}, R_{01}, R_{10}, R_{11})$, where $R_{ab} = E_2^2 \setminus \{(a,b)\}$, is in P. 
\end{example}

\begin{example} \label{ex:col}
Let $\relstr{K}_3 = (E_3; N_3)$, where
 $N_3$ is the binary inequality relation, i.e., 
$$N_3 = \{(a,b) \in E_3^2: a \neq b\}.$$
An input to $\CSP(\relstr{K}_3)$ is, e.g.,
$$
(\exists x_1 \dots \exists x_5) \ N_3(x_1, x_2) \wedge N_3(x_1, x_3) \wedge N_3(x_1, x_4) \wedge N_3(x_2,x_3) \wedge N_3(x_2,x_4).
$$
Here an input can be drawn as a graph -- its vertices are the variables and vertices $x,y$ are declared adjacent iff the input contains a conjunct $N_3(x,y)$ or $N_3(y,x)$. For example, the graph associated to the input above is the five vertex graph obtained by merging two triangles along an edge. Clearly, an input sentence is true if and only if the vertices of the associated graph can be colored by colors 0, 1, and 2 so that adjacent vertices receive different colors. Therefore $\CSP(\relstr{K}_3)$ is essentially the same as the 3-coloring problem for graphs, another well known NP-complete problem. 

More generally, $\CSP(\relstr{K}_k) = (E_k,N_k)$, where $N_k$ is the inequality relation on $E_k$, is NP-complete for $k \geq 3$ and in P for $k=2$. 
\end{example}

\begin{example} \label{ex:nae}
Let $3\relstr{NAE}_k = (E_k; 3NAE_k)$, where
 $3NAE_k$ is the ternary not-all-equal relation, i.e., 
$$
3NAE_k = E_k^3 \setminus \{(a,a,a): a \in E_k\}.
$$
Taking the viewpoint of Example~\ref{ex:sat}, the CSP over $3\relstr{NAE}_2$ is the positive not-all-equal 3SAT, 
where one is given a 3SAT formula without negations and the aim is to decide whether there is an assignment such that, in every clause, not all variables get the same value. This problem is NP-complete~\cite{Sch78}. 

From the graph theoretical viewpoint, $\CSP(3\relstr{NAE}_k)$ is the problem of deciding whether a given 3-uniform hypergraph%
\footnote{Here we should rather say a hypergraph whose hyperedges have size at most 3 because of conjuncts of the form $3NAE_k(x,x,y)$ or $3NAE_k(x,x,x)$. Let us ignore this minor technical imprecision.} 
admits a coloring by $k$ colors so that no hyperedge is monochromatic. 
\end{example}

\begin{example} \label{ex:one}
Let $1\relstr{IN}3 = (E_2; 1IN3)$, where
$$
1IN3 = \{(1,0,0),(0,1,0),(1,0,0)\}.
$$
The CSP over $1\relstr{IN}3$ is the positive one-in-three SAT problem or, in other words, the problem of deciding whether a given 3-uniform hypergraph admits a coloring by colors $0$ and $1$ so that exactly one vertex in each hyperedge receives the color 1. This problem is, again, NP-complete~\cite{Sch78}. 
\end{example}

\begin{example} \label{ex:lin}
Let $3\relstr{LIN}_5 = (E_5; L_{0000}, L_{0001}, \dots, L_{4444})$, where
$$
L_{abcd} = \{(x,y,z) \in E_5^3: ax+by+cz=d \pmod 5\}.
$$
An input, such as
$$
(\exists x_1 \exists x_2 \dots) \ L_{1234}(x_3,x_4,x_2) \wedge L_{4321}(x_5,x_1,x_3) \wedge \dots
$$
can be written as a system of linear equations over the 5-element field $\mathbb{Z}_5$, such as
$$
1x_3 + 2x_4 + 3x_2 = 4, \ 4x_5 + 3x_1 + 2x_3 = 1, \ \dots,
$$
therefore $\CSP(3\relstr{LIN}_5)$ is essentially the same problem as deciding whether a system of linear equations over $\mathbb{Z}_5$ (with each equation containing 3 variables) has a solution. This problem is in P. 
\end{example}

\subsection{1st step: polymorphisms}

The crucial concept for the algebraic approach to the CSP is a polymorphism, which is a homomorphism from a cartesian power of a structure to the structure:

\begin{definition}
Let $\relstr{A} = (A; R_1, \dots, R_n)$ be a relational structure. A $k$-ary (total) function $f: A^k \to A$ is a \emph{polymorphism} of $\relstr{A}$ if it is compatible with every relation $R_i$, that is, for all tuples $\vc{r}_1, \dots, \vc{r}_k \in R_i$, the tuple $f(\vc{r}_1, \dots, \vc{r}_k)$ (where $f$ is applied component-wise) is in $R_i$. 

By $\Pol(\relstr{A})$ we denote the set of all polymorphisms of $\relstr{A}$. 
\end{definition}

The compatibility condition is often stated as follows: for any $(\arity{R_i} \times k)$-matrix whose column vectors are in $R_i$, the vector obtained by applying $f$ to its rows is in $R_i$ as well. 

Note that the unary polymorphisms of $\relstr{A}$ are exactly the endomorphisms of $\relstr{A}$. One often thinks of endomorphisms (or just automorphisms) as symmetries of the structure. In this sense, polymorphisms can be thought of as higher dimensional symmetries. 

For any domain $A$ and any $i \leq k$, the $k$-ary projection to the $i$-th coordinate, that is, the function $\pi^k_i: A^k \to A$ defined by
$$
\pi^k_i (x_1, \dots, x_n) = x_i,
$$
is a polymorphism of every structure with domain $A$. These are the \emph{trivial} polymorphisms. 
The following examples show some nontrivial polymorphisms.

\begin{example} \label{ex:maj}
Consider the template $2\relstr{SAT}$ from Example~\ref{ex:sat}.
It is easy to verify that the ternary majority function $\maj: E_2^3 \to E_2$ given by
$$
\maj(x,x,y) = \maj(x,y,x) = \maj(y,x,x) = x \quad \mbox{for all $x,y \in E_2$}
$$
is a polymorphism of $2\relstr{SAT}$.

In fact, whenever a relation $R \subseteq E_2^m$ is compatible with $\maj$, it can be pp-defined (that is, defined by a pp-formula) 
from relations in $2\relstr{SAT}$ (see e.g.~\cite{JCG97}). Now for any template $\relstr{A} = (E_2; R_1, \dots, R_n)$ with polymorphism $\maj$, an input of $\CSP(\relstr{A})$ can be easily rewritten to an equivalent input of $\CSP(2\relstr{SAT})$ and therefore $\CSP(\relstr{A})$ is in P.
\end{example}

\begin{example}
Consider the template $3\relstr{LIN}_5$ from Example~\ref{ex:lin}.
Each relation in this structure is an affine subspace of $\mathbb{Z}_5^3$. Every affine subspace is closed under affine combinations, therefore, for every $k$ and every $t_1, \dots, t_k \in E_5$ such that $t_1 + \dots + t_k=1 \pmod{5}$, the $k$-ary function $f_{t_1, \dots, t_k}: E_5^k \to E_5$ defined by
$$
f_{t_1, \dots, t_k}(x_1, \dots, x_k) = t_1x_1 + \dots, t_kx_k \pmod{5}
$$
is a polymorphism of $3\relstr{LIN}_5$. 

Conversely, every subset of $A^m$ closed under affine combinations is an affine subspace of $\mathbb{Z}_5^m$. It follows that if every $f_{t_1, \dots, t_k}$ is a polymorphism of  $\relstr{A} = (E_5; R_1, \dots, R_n)$, then inputs to $\CSP(\relstr{A})$ can be rewritten to systems of linear equations over $\mathbb{Z}_5$ and thus $\CSP(\relstr{A})$ is in P.  
\end{example}

The above examples also illustrate that polymorphisms influence the computational complexity. 
The first step of the algebraic approach was to realize that this is by no means a coincidence.

\begin{theorem}[\cite{Jea98}] \label{thm:first}
The complexity of $\CSP(\relstr{A})$ depends only on $\Pol(\relstr{A})$.

More precisely, if $\relstr{A}$ and $\relstr{B}$ are finite relational structures and $\Pol(\relstr{A}) \subseteq \Pol(\relstr{B})$, then $\CSP(\relstr{B})$ is (log-space) reducible to $\CSP(\relstr{A})$. In particular, if $\Pol(\relstr{A})=\Pol(\relstr{B})$, then $\CSP(\relstr{A})$ and $\CSP(\relstr{B})$ have the same complexity.
\end{theorem}

\begin{proof}[sketch]
If $\Pol(\relstr{A}) \subseteq \Pol(\relstr{B})$, then relations in $\relstr{B}$ can be pp-defined from relations in $\relstr{A}$ by a classical result in Universal Algebra~\cite{Gei68,BKKR69,BKKR69a}. This gives a reduction from $\CSP(\relstr{B})$ to $\CSP(\relstr{A})$.
\end{proof}

Theorem~\ref{thm:first} can be used as a tool for proving NP-hardness: when $\relstr{A}$ has only trivial polymorphism (and has domain of size at least two), any CSP on the same domain can be reduced to $\CSP(\relstr{A})$ and therefore $\CSP(\relstr{A})$ is NP-complete. This NP-hardness criterion is not perfect, e.g., $\CSP(3\relstr{NAE}_2)$ has a nontrivial endomorphism $x \mapsto 1-x$. 

\subsection{2nd step: strong Maltsev conditions}

Theorem~\ref{thm:first} shows that the set of polymorphisms determines the complexity of a CSP. What information do we really need to know about the  polymorphisms to determine the complexity? It has turned out that it is sufficient to know which functional equations they solve. 
In the following definition we use a standard universal algebraic term for a functional equation, a strong Maltsev condition.

\begin{definition}
A \emph{strong Maltsev condition} over a set of function symbols $\Sigma$ is a finite set of equations of the form $t=s$, where $t$ and $s$ are terms built from variables and symbols in $\Sigma$. 	

Let $\clone{M}$ be a set of functions on a common domain. A strong Maltsev condition $S$ is \emph{satisfied} in $\clone{M}$  if the function symbols of $\Sigma$ can be interpreted in $\clone{M}$ so that each equation in $S$ is satisfied for every choice of variables. 
\end{definition}

\begin{example} \label{ex:random_maltsev}
A strong Maltsev condition over $\Sigma = \{f,g,h\}$ (where $f$ and $g$ are binary symbols and $h$ is ternary) is, e.g.,
\begin{align*}
f(g(f(x,y),y),z) &= g(x,h(y,y,z)) \\
f(x,y) &= g(g(x,y),x).
\end{align*}
This condition is satisfied in the set of all projections (on any domain) since, by interpreting $f$ and $g$ as $\pi^2_1$ and $h$ as $\pi^3_1$, both equations are satisfied for every $x,y,z$ in the domain -- they are equal to $x$. 
\end{example}

The strong Maltsev condition in the above example is not interesting for us since it is satisfied in every $\Pol(\relstr{A})$.
Such conditions are called trivial:

\begin{definition}
A strong Maltsev condition is called \emph{trivial} if it is satisfied in the set of all projections on a two-element set (equivalently, it is satisfied in $\Pol(\relstr{A})$ for every $\relstr{A}$).
\end{definition}

Two nontrivial Maltsev condition are shown in the following example. 

\begin{example} \label{ex:maltsev}
The strong Maltsev condition (over a single ternary symbol $m$)
\begin{align*}
m(x,x,y) &= x \\
m(x,y,x) &= x \\
m(y,x,x) &= x
\end{align*}
is nontrivial since each of the possible interpretations $\pi^3_1$, $\pi^3_2$, $\pi^3_3$ of $m$ falsifies one of the equations.
This condition is satisfied in $\Pol(2\relstr{SAT})$ by interpreting $m$ as the majority function, see Example~\ref{ex:maj}.

The strong Maltsev condition
\begin{align*}
p(x,x,y) &= y \\
p(y,x,x) &= y
\end{align*}
is also nontrivial. It is satisfied in $\Pol(3\relstr{LIN}_5)$ by interpreting $p$ as $x+4y+z \pmod 5$.

In fact, if $\Pol(\relstr{A})$ satisfies one of the strong Maltsev conditions in this example, then $\CSP(\relstr{A})$ is in P (see e.g.~\cite{BKW17}).
\end{example}

The following theorem is (a restatement of) the second crucial step of the algebraic approach.

\begin{theorem}[\cite{BJK05}, see also \cite{Bod08}] \label{thm:second}
The complexity of $\CSP(\relstr{A})$ depends only on strong Maltsev conditions satisfied by $\Pol(\relstr{A})$.

More precisely, if $\relstr{A}$ and $\relstr{B}$ are finite relational structures and each strong Maltsev condition satisfied in $\Pol(\relstr{A})$ is satisfied in $\Pol(\relstr{B})$, then $\CSP(\relstr{B})$ is (log-space) reducible to $\CSP(\relstr{A})$. In particular, if $\Pol(\relstr{A})$ and $\Pol(\relstr{B})$ satisfy the same strong Maltsev conditions, then $\CSP(\relstr{A})$ and $\CSP(\relstr{B})$ have the same complexity.
\end{theorem}

\begin{proof}[sketch]
The proof can be done in a similar way as for Theorem~\ref{thm:first}. Instead of pp-definitions one uses more general constructions called pp-interpretations and, on the algebraic side, the Birkhoff HSP theorem~\cite{Birk35}.
\end{proof}

Theorem~\ref{thm:second} gives us an improved tool for proving NP-hardness: if $\Pol(\relstr{A})$ satisfies only trivial strong Maltsev conditions, then $\CSP(\relstr{A})$ is NP-hard. This criterion is better, e.g., it can be applied to $\CSP(3\relstr{NAE}_2)$, but still not perfect, e.g., it cannot be applied to the CSP over the disjoint union of two copies of $\relstr{K}_3$.

\subsection{3rd step: minor conditions}

Strong Maltsev conditions that appear naturally in the CSP theory or in Universal Algebra are often of an especially simple form, they involve no nesting of function symbols. The third step in the basics of the algebraic theory was to realize that this is also not a coincidence.

\begin{definition}
A strong Maltsev condition is called a \emph{minor condition} if each side of every equation contains exactly one occurrence of a function symbol. 
\end{definition}

In other words, each equation in a strong Maltsev condition is of the form ``symbol(variables) = symbol(variables)''.

\begin{example}
The condition in Example~\ref{ex:random_maltsev} is not a minor condition since, e.g., the left-hand side of the first equation involves three occurrences of function symbols.

The conditions in Example~\ref{ex:maltsev} are not minor conditions either since the right-hand sides do not contain any function symbol. However, these conditions have close friends which are minor conditions. For instance, the friend of the second system is the minor condition
\begin{align*}
p(x,x,y) &= p(y,y,y) \\
p(y,x,x) &= p(y,y,y).
\end{align*}
Note that this system is also satisfied in $\Pol(3\relstr{LIN}_5)$  by the same interpretation as in Example~\ref{ex:maltsev}, that is, $x+4y+z \pmod{5}$.
\end{example}

The following theorem is a strengthening of Theorem~\ref{thm:second}. We give only the informal part of the statement, the precise formulation is analogous to Theorem~\ref{thm:second}. 

\begin{theorem}[\cite{BOP18}] \label{thm:third}
The complexity of $\CSP(\relstr{A})$ (for finite $\relstr{A}$) depends only on minor conditions satisfied by $\Pol(\relstr{A})$.
\end{theorem}

\begin{proof}[sketch]
The proof again follows the same pattern by further generalizing pp-interpretations (to so called pp-constructions) and the Birkhoff HSP theorem.
\end{proof}

\subsection{Classification}

Just like Theorems~\ref{thm:first} and \ref{thm:second} give hardness criteria for CSPs, we get an improved sufficient condition for NP-hardness as a corollary of Theorem~\ref{thm:third}.

\begin{corollary}
Let $\relstr{A}$ be a finite relational structure which satisfies only trivial minor conditions. Then $\CSP(\relstr{A})$ is NP-complete.
\end{corollary}

Bulatov, Jeavons, and Krokhin have conjectured~\cite{BJK05} that satisfying only trivial minor conditions is actually the only reason for hardness%
\footnote{Their conjecture is equivalent but was, of course, originally stated in a different language -- the significance of minor conditions in CSPs was identified much later.}. 
Intensive efforts to prove this conjecture, called the \emph{tractability conjecture} or the \emph{algebraic dichotomy conjecture}, have recently culminated in two independent proofs by Bulatov and Zhuk:

\begin{theorem}[\cite{Bul17},\cite{Zhu17}]
If a finite relational structure $\relstr{A}$ satisfies a nontrivial minor condition, then $\CSP(\relstr{A})$ is in P.
\end{theorem}

Thus we now have a complete classification result: every finite structure either satisfies a nontrivial minor condition and then its CSP is in P, or it does not and its CSP is NP-complete.
The proofs of Bulatov and Zhuk are very complicated and it should be stressed out that the basic steps presented in this paper form only a tiny (but important) part of the theory. 

In fact, the third step did not impact on the resolution of the tractability conjecture for CSP over finite domains at all. However, it turned out to be significant for some generalizations of the CSP, including the generalization that we discuss in the next section, the Promise CSP.

\section{PCSP}

Many fixed-template CSPs, such as finding a 3-coloring of a graph or finding a satisfying assignment to a 3SAT formula, are hard computational problems. There are two ways how to relax the requirement on the
assignment in order to get a potentially simpler problem. The
first one is to require a specified fraction of the constraints to
be satisfied. For example, given a satisfiable 3SAT input,
is it easier to find an assignment satisfying at least $90\%$ of 
clauses? A celebrated result of H{\aa}stad \cite{Hst01}, which strengthens the famous PCP Theorem~\cite{ALMSS98,AS98}, 
proves that the answer is negative -- it is still an NP-complete problem. (Actually,
any fraction greater than 7/8 gives rise to an NP-complete
problem while the fraction 7/8 is achievable in polynomial
time.)

The second type of relaxation, the one we consider in this paper, is to require that a specified
weaker version of every constraint is satisfied. For example,
we want to find a 100-coloring of a 3-colorable graph, or
we want to find a valid $\CSP(3\relstr{NAE}_2)$ assignment to a true input of $\CSP(1\relstr{IN}3)$. 
This idea is formalized in the following definition.    

\begin{definition} \label{def:PCSP}
Let $\relstr{A} = (A; R_1^{\relstr{A}}, R_2^{\relstr{A}}, \dots, R_n^{\relstr{A}})$ 
and $\relstr{B} = (B; R_1^{\relstr{B}}, R_2^{\relstr{B}}, \dots, R_n^{\relstr{B}})$ be two similar finite relational structures (that is, $R^{\relstr{A}}$ and $R^{\relstr{B}}$ have the same arity for each $i$), and assume that there exists a homomorphism $\relstr{A} \to \relstr{B}$. Such a pair $(\relstr{A},\relstr{B})$ is refered to as a \emph{PCSP template}.

The PCSP over $(\relstr{A},\relstr{B})$, denoted $\PCSP(\relstr{A},\relstr{B})$, is the problem to distinguish, given a pp-sentence $\phi$ over the relational symbols $R_1, \dots, R_n$, between the cases that $\phi$ is true in $\relstr{A}$ (answer ``Yes'') and $\phi$ is not true in $\relstr{B}$ (answer ``No''). 
\end{definition}

For example, consider $\relstr{A} = (A; R^{\relstr{A}},S^{\relstr{A}})$ and $\relstr{B} = (B; R^{\relstr{B}},S^{\relstr{B}})$, where all the relations are binary. An input to $\PCSP(\relstr{A},\relstr{B})$ is, e.g., 
\[
(\exists x_1 \exists x_2 \dots \exists x_5) \ R(x_1,x_3) \wedge S(x_5,x_2) \wedge R(x_3,x_3).
\]
The algorithm should answer ``Yes'' if the sentence is true in $\relstr{A}$, i.e., the following sentence is true 
\[
(\exists x_1 \exists x_2 \dots \exists x_5) \ R^{\relstr{A}}(x_1,x_3) \wedge S^{\relstr{A}}(x_5,x_2) \wedge R^{\relstr{A}}(x_3,x_3),
\]
and the algorithm should answer ``No'' if the sentence is not true in $\relstr{B}$. In case that neither of the cases takes place, we do not have any requirements on the algorithm. Alternatively, we can say that the algorithm is promised that the input satisfies either ``Yes'' or ``No'' and it is required to decide which of these two cases takes place. 

Note that the assumption that $\relstr{A} \to \relstr{B}$ is necessary for the problem to make sense, otherwise, the  ``Yes'' and ``No'' cases would not be disjoint. Also observe that $\CSP(\relstr{A})$ is the same problem as $\PCSP(\relstr{A},\relstr{A})$. 

The search version of $\PCSP(\relstr{A},\relstr{B})$ is perhaps a bit more natural problem: the goal is to find a $\relstr{B}$-satisfying assignment given an $\relstr{A}$-satisfiable input. Unlike in the CSP, it is not known whether the search version can be harder than the decision version presented in Definition~\ref{def:PCSP}.

\subsection{Examples}

The examples below show that PCSPs are richer than CSP, both on the algorithmic and the hardness side. 

\begin{example} \label{ex:aprox_col}
Recall the structure $\relstr{K}_k$ from Example~\ref{ex:col} consisting of the inequality relation on a $k$-element set. 
For $k \leq l$, the PCSP over $(\relstr{K}_k,\relstr{K}_l)$ is the problem to distinguish between $k$-colorable graphs and graphs that are not even $l$-colorable (or, in the search version, the problem to find an $l$-coloring of a $k$-colorable graph). 

Unlike for the case $k=l$, the complexity of this problem for $3 \leq k < l$ is a notorious open question. It is conjectured that $\PCSP(\relstr{K}_k,\relstr{K}_l)$ is NP-hard for every $k<l$, but this conjecture was confirmed only in special cases: for $l \leq 2k-2$~\cite{BG16} (e.g., 4-coloring a 3-colorable graph) and for a large enough $k$ and $l \leq 2^{\Omega(k^{1/3})}$~\cite{Hua13}. 
The algebraic development discussed in the next subsection helped in improving the former result to $l \leq 2k-1$~\cite{BKO19} (e.g., 5-coloring a 3-colorable graph).
\end{example}

\begin{example} \label{ex:aprox_nae}
Recall the structure $3\relstr{NAE}_k$ from Example~\ref{ex:nae} consisting of the ternary not-all-equal relation on a $k$-element set.
For $k \leq l$, the PCSP over $(3\relstr{NAE}_k, 3\relstr{NAE}_l)$ is essentially the problem to distinguish between $k$-colorable 3-uniform hypergraphs and 3-uniform hypergraphs that are not even $l$-colorable. 

This problem is NP-hard for every $2 \leq k \leq l$~\cite{DRS05}, the proof uses strong tools, the PCP theorem and Lov\'asz's theorem on the chromatic number of Kneser's graphs~\cite{Lov78}.
\end{example}

\begin{example} \label{ex:one_nae}
Recall from Example~\ref{ex:one} that $1\relstr{IN}3$ denotes the structure on the domain $E_2$ with the ternary ``one-in-three'' relation $1IN3$. The PCSP over $(1\relstr{IN}3, 3\relstr{NAE}_2)$ is the problem to distinguish between $3$-uniform hypergraphs, which admit a coloring by colors $0$ and $1$ so that exactly one vertex in each hyperedge receives the color 1, and $3$-uniform hypergraphs that are not even $2$-colorable.

This problem, even its search version, admits elegant polynomial time algorithms~\cite{BG18,BG19} -- one is based on solving linear equations over the integers, another one on linear programming. For this specific template, the algorithm can be further simplified as follows.

We are given a 3-uniform hypergraph, which admits a coloring by colors $0$ and $1$ so that $(x,y,z) \in 1IN3$ for every hyperedge $\{x,y,z\}$, and we want to find a 2-coloring. We create a system of linear equations \emph{over the rationals} as follows: for each hyperedge $\{x,y,z\}$ we introduce the equation $x+y+z = 1$. By the assumption on the input hypergraph, this system is solvable in $\{0,1\} \subseteq \mathbb{Q}$ (in fact, $\{0,1\}$-solutions are the same as valid $1IN3$-assignements). Solving equations in $\{0,1\}$ is hard, but it is possible to solve the system in $\mathbb{Q} \setminus \{1/3\}$ in polynomial time by a simple adjustment of Gaussian elimination. Now we assign $1$ to a vertex $x$ if $x > 1/3$ in our rational solution, and $0$ otherwise. It is simple to see that we get a valid $2$-coloring.

Interestingly, to solve $\PCSP(1\relstr{IN}3, 3\relstr{NAE}_2)$, the presented algorithm uses a CSP over an \emph{infinite} structure, namely $(\mathbb{Q} \setminus \{1/3\};R)$, where $R = \{(x,y,z) \in (\mathbb{Q} \setminus \{1/3\})^3: x+y+z=1\}$. In fact,  the  infinity is necessary for this PCSP, see~\cite{Bar19} for a formal statement and a proof. 
\end{example}

\subsection{4th step: minor conditions!}

After the introduction of the PCSP framework, it has quickly turned out that both the notion of a polymorphism and Theorem~\ref{thm:first} have straightforward generalizations.

\begin{definition}
Let $\relstr{A} = (A; R_1^{\relstr{A}}, \dots)$ and $\relstr{B} = (B; R_1^{\relstr{B}}, \dots)$ be two similar relational structures. A $k$-ary (total) function $f: A^k \to B$ is a \emph{polymorphism} of $(\relstr{A},\relstr{B})$ if it is compatible with every pair  $(R_i^{\relstr{A}},R_i^{\relstr{B}})$, that is, for all tuples $\vc{r}_1, \dots, \vc{r}_k \in R_i^{\relstr{A}}$, the tuple $f(\vc{r}_1, \dots, \vc{r}_k)$  is in $R_i^{\relstr{B}}$. 

By $\Pol(\relstr{A},\relstr{B})$ we denote the set of all polymorphisms of $(\relstr{A},\relstr{B})$. 
\end{definition}

\begin{example}
For every $k$ which is not disible by 3, the  $k$-ary ``1/3-threshold'' function $f: E_2^k \to E_2$ defined by
$$
f(x_1, \dots, x_k) = \left\{ \begin{array}{ll}1 & \mbox{if} \sum x_i / k > 1/3\\ 0 & \mbox{else}\end{array}\right.
$$
is a polymorphism of the PCSP template $(1\relstr{IN}3, 3\relstr{NAE}_2)$ from Example~\ref{ex:one_nae}. Any PCSP whose template (over the domains $E_2$ and $E_2$) admits all these polymorphisms is in P~\cite{BG18,BG19}. 
\end{example}

\begin{theorem}[\cite{BG16a}]
The complexity of $\PCSP(\relstr{A},\relstr{B})$ depends only on $\Pol(\relstr{A},\relstr{B})$.
\end{theorem}

\begin{proof}[sketch]
Proof is similar to Theorem~\ref{thm:first} using~\cite{Pip02} instead of~\cite{Gei68,BKKR69,BKKR69a}.
\end{proof}

Note that, in general, composition of polymorphisms is not even well-defined. Therefore the second step, considering strong Maltsev conditions satisfied by polymorphisms, does not make sense for PCSPs. However, minor conditions make perfect sense and they do capture the complexity of PCSPs, as proved in~\cite{BKO19}. Furthermore, the paper~\cite{BKO19} also provides an alternative proof by directly relating a PCSP to a computational problem concerning minor conditions!

\begin{theorem}[\cite{BKO19}] \label{thm:forth}
Let $(\relstr{A},\relstr{B})$ be a PCSP template and $\clone{M} = \Pol(\relstr{A},\relstr{B})$. 
The following computational problems are equivalent for every sufficiently large $N$.
\begin{itemize}
\item $\PCSP(\relstr{A},\relstr{B})$.
\item Distinguish, given a minor condition $\mc{C}$ whose function symbols have arity at most $N$, 
between the cases that $\mc{C}$ is trivial and $\mc{C}$ is not satisfied in $\clone{M}$.
\end{itemize}
\end{theorem}

\begin{proof}[sketch]
The reduction from $\PCSP(\relstr{A},\relstr{B})$ to the second problem works as follows. Given an input to the PCSP we introduce one $|A|$-ary function symbol $g_a$ for each variable $a$ and one $|R^{\relstr{A}}|$-ary function symbol $f_C$ for each conjunct $R(\dots)$. The way to build a minor condition is quite natural, for example, the input
$$
(\exists a \exists b \exists c \exists d) \  R(c,a,b) \wedge R(a,d,c) \wedge \dots
$$
to $\PCSP(1\relstr{IN}3, 3\relstr{NAE}_2)$ is transformed to the minor condition
\begin{align*}
f_1(x_1,x_0,x_0) &= g_c(x_0,x_1) \\
f_1(x_0,x_1,x_0) &= g_a(x_0,x_1) \\
f_1(x_0,x_0,x_1) &= g_b(x_0,x_1) \\
\smallskip \\
f_2(x_1,x_0,x_0) &= g_a(x_0,x_1) \\
f_2(x_0,x_1,x_0) &= g_d(x_0,x_1) \\
f_2(x_0,x_0,x_1) &= g_c(x_0,x_1) \\
\dots
\end{align*}
It is easy to see that a sentence that is true in $\relstr{A}$ is transformed to a trivial minor condition.
On the other hand, if the minor condition is satisfied in $\clone{M}$, say by the functions denoted $f'_1,f'_2, g'_a, \dots$, then the mapping
$a \mapsto g'_a(0,1)$, $b \mapsto g'_b(0,1)$, \dots gives a $\relstr{B}$-satisfying assignment of the sentence -- this can be deduced from the fact that $f'$s and $g'$ are polymorphisms.

The reduction in the other direction is based on the idea that the question ``Is this minor condition satisfied by polymorphisms of $\relstr{A}$?'' can be interpreted as an input to $\CSP(\relstr{A})$. The main ingredient is to look at functions as tuples (their tables); then ``$f$ is a polymorphism'' translates to a conjunction, and equations can be simulated by merging variables. 
\end{proof}

Theorem~\ref{thm:forth} implies Theorem~\ref{thm:third} (and its generalization to PCSPs) since
the computational problem in the second item clearly only depends on minor conditions satisfied in $\clone{M}$. The proof sketched above 
\begin{itemize}
\item is simple and does not (explicitly) use any other  results, such as the correspondence between polymorphisms and pp-definitions used in Theorem~\ref{thm:first} or the Birkhoff HSP theorem used in Theorem~\ref{thm:second}, and
\item is based on constructions which have already appeared, in some form, in several contexts; in particular, the second item is related to important problems in approximation, versions of the Label Cover problem (see~\cite{BKO19,BBKO}). 
\end{itemize}
The theorem and its proof are simple, nevertheless, very useful. For example, the hardness of $\PCSP(\relstr{K}_k,\relstr{K}_{2k-1})$ mentioned in Example~\ref{ex:aprox_col} was proved in~\cite{BKO19} by showing that every minor condition satisfied in $\Pol(\relstr{K}_k,\relstr{K}_{2k-1})$ is satisfied in $\Pol(3\relstr{NAE}_2,3\relstr{NAE}_l)$ (for some $l$) and then using the NP-hardness of $\PCSP(3\relstr{NAE}_2,3\relstr{NAE}_l)$ proved in~\cite{DRS05} (see Example~\ref{ex:aprox_nae}).

\section{Conclusion}

The PCSP framework is much richer than the CSP framework; on the other hand, the basics of the algebraic theory generalize from CSP to PCSP, as shown in Theorem~\ref{thm:forth}. Strikingly, the computational problems in Theorem~\ref{thm:forth} are (promise and restricted versions) of two ``similar'' problems:
\begin{itemize}
\item[(i)] Given a structure $\struct{A}$ and a first-order sentence $\phi$ over the same signature, decide whether $\struct{A}$ satisfies $\phi$.
\item[(ii)] Given a structure $\struct{A}$ and a first-order sentence $\phi$ in a different signature, decide whether symbols in $\phi$ can be interpreted in $\struct{A}$ so that $\struct{A}$ satisfies $\phi$.
\end{itemize}
Indeed, $\CSP(\relstr{A})$ is the problem (i) with $\struct{A}$ a fixed relational structure and $\phi$ a pp-sentence (and PCSP is a promise version of this problem), whereas a promise version of the problem (ii) restricted to a fixed $\struct{A}$ of purely functional signature and universally quantified conjunctive first-order sentences $\phi$ is the second item in Theorem~\ref{thm:forth}. Variants of problem (i) appear in many contexts throughout Computer Science. What about problem (ii)?

%
%
%

\bibliographystyle{plain}
\bibliography{Libor_FCT}

\end{document}